\documentclass[a4paper,UKenglish,cleveref, autoref]{lipics-v2019}

\newcommand{\zfps}{\mathbb{Z}\langle \langle A \rangle \rangle} 

\usepackage{tikz}
\usetikzlibrary{automata, positioning, arrows}

\newcommand{\N}{{\mathbb{N}}}  
\newcommand{\Z}{{\mathbb{Z}}}  
\newcommand{\Q}{{\mathbb{Q}}} 
\newcommand{\F}{{\mathbb{F}}}  

\newcommand{\0}{\mathbf{0}} 
\newcommand{\pfa}{\mathcal{P}}
\newcommand{\alp}[1]{\Sigma_{#1}}

\theoremstyle{plain}
\newtheorem{prob}[theorem]{Problem}

\title{Polynomially Ambiguous Probabilistic Automata on Restricted Languages} 

\author{Paul C. Bell}{Department of Computer Science, Byrom Street, Liverpool~John~Moores~University, Liverpool, L3-3AF, UK}{p.c.bell@ljmu.ac.uk}{ https://orcid.org/0000-0003-2620-635X}{}

\authorrunning{P.\,C. Bell}

\Copyright{Paul Charles Bell}



\relatedversion{A preliminary version of this manuscript was presented at \cite{Be19}.}

\nolinenumbers 

\hideLIPIcs  

\begin{document}

\maketitle

\begin{abstract} We consider the computability and complexity of decision questions for Probabilistic Finite Automata (PFA) with sub-exponential ambiguity. We show that the emptiness problem for strict and non-strict cut-points of polynomially ambiguous commutative PFA remains undecidable, implying that the problem is undecidable when inputs are from a letter monotonic language. We show that the problem remains undecidable over a binary input alphabet when the input word is over a bounded language, in the noncommutative case. In doing so, we introduce a new technique based upon the Turakainen construction of a PFA from a Weighted Finite Automata which can be used to generate PFA of lower dimensions and of subexponential ambiguity. We also study freeness/injectivity problems for polynomially ambiguous PFA and study the border of decidability and tractability for various cases. 
\end{abstract}

\section{Introduction}
Probabilistic Finite Automata (PFA) are a simple yet expressive model of computation, obtained by extending Nondeterministic Finite Automata (NFA) so that transitions  from each state (and for each input letter) form probability distributions. As input letters are read from some alphabet $\Sigma$, the automaton transitions among states according to these probabilities. The probability of a PFA $\pfa$ accepting a word $w \in \Sigma^*$ is given by the probability of the automaton being in one of its final states, denoted $f_{\mathcal{P}}(w) = {\bf x}^TM_{w_1}M_{w_{2}} \cdots M_{w_k} {\bf y}$, where ${\bf x}$ represents the initial state, ${\bf y}$ represents the final state and each $M_{w_i}$ is a row stochastic matrix representing the transition probabilities for letter $w_i \in \Sigma$.

The PFA model has been studied extensively over the years, ever since its introduction by Rabin \cite{Ra63}; for example see \cite{Bu80} for a survey of $416$ research papers related to PFA in the eleven years since their introduction to just 1974. They have been used to study Arthur-Merlin games \cite{BM88}, space bounded interactive proofs \cite{CJ89}, quantum complexity theory \cite{YS11}, the joint spectral radius and semigroup boundedness \cite{BT00}, Markov decision processes and planning questions \cite{BT00b}, and text and speech processing \cite{MPR02} among many other applications.

There are a variety of interesting questions that one may ask about PFA. A central question is the \emph{emptiness problem} for cut-point languages; given some probability $\lambda \in [0, 1]$, does there exist a finite input word whose probability of acceptance is greater than $\lambda$ (i.e. does there exist $w \in \Sigma^*$ such that $f_{\mathcal{P}}(w) > \lambda$, see Section~\ref{pfa-def-sec}). This problem is known to be undecidable \cite{Paz}, even for a fixed number of dimensions and for two input matrices \cite{BC03, Hir}. A second natural question is the \emph{freeness problem} (or \emph{injectivity problem}) for PFA, studied in \cite{BCJ16} - given a PFA $\pfa$ over alphabet $\Sigma$ determine whether the acceptance function $f_{\mathcal{P}}(w)$ is injective (i.e. do there exist two distinct words with the same acceptance probability). 

When studying the frontiers of decidability of a problem, there are two competing objectives, namely, determine the most general version of the problem which is decidable, and the most restricted specialization which is undecidable; the latter being the main focus of this paper. 

Various classes of restrictions may be studied for PFA, depending upon the structure of the PFA or on possible input words. Some restrictions relate to the number of states of the automaton, the alphabet size and whether one defined the PFA over the algebraic real numbers or the rationals. One may also study PFA with finite, polynomial or exponential ambiguity (in terms of the underlying NFA), PFA defined for restricted input words (for example those coming from regular, bounded or letter monotonic languages), PFA with isolated thresholds (a probability threshold is isolated if it cannot be approached arbitrarily closely) and commutative PFA, where all transition matrices commute, for which cut-point languages and non-free languages generated by such automata necessarily become commutative.

The cut-point emptiness problem for PFA is known to be undecidable for rational matrices \cite{Paz}, even over a binary alphabet when the PFA has dimension $46$ in \cite{BC03}; later improved to dimension $25$ \cite{Hir}. The authors of \cite{BMT77} show that the problem of determining if a threshold is isolated (resp. if a PFA has any isolated threshold) is undecidable and this was shown to hold even for PFA with $420$ (resp. $2354$) states over a binary alphabet \cite{BC03}.

A natural restriction on PFA was studied in \cite{BHH13}, where possible input words of the PFA are restricted to be from some letter monotonic language of the form $\mathcal{L} = a_1^*a_2^* \cdots a_k^*$ with each $a_i \in \Sigma$ (analogous to a 1.5 way PFA, whose read head may ``stay put'' on an input word letter but never moves left). In other words, we ask if there exists some  $w \in \mathcal{L}$ such that $f_{\mathcal{P}}(w) > \lambda$. This restriction is inspired by the well-known property that many language-theoretic problems become decidable or tractable when restricted to bounded languages, and especially letter monotonic languages \cite{CH14}. Nevertheless, the emptiness problem for PFA on letter monotonic languages was shown to be undecidable for high (but finite) dimensional matrices over the rationals via an encoding of Hilbert's tenth problem on the solvability of Diophantine equations and the utilization of Turakainen's method to transform weighted integer automata to probabilistic automata \cite{Tu69}. 

The authors of \cite{FR17} recently studied decision problems for PFA of various degrees of ambiguity in order to map the frontier of decidability for restricted classes of PFA. The degree of ambiguity of a PFA is a structural property, giving an indication of the number of accepting runs for a given input word and it can be used to give various classifications of ambiguity including finite, polynomial and exponential ambiguity (formal details are given in Section~\ref{amb-sec}). The ambiguity of a PFA is a property of the underlying NFA and is independent of the transition probabilities in so much as we only need care if the probability is zero or positive. The degree of ambiguity of automata is a well-known and well-studied property in automata theory \cite{WS91}. The authors of \cite{FR17} show that the emptiness problem for PFA remains undecidable even for polynomially ambiguous automata (quadratic ambiguity), before going on to show {\bf PSPACE}-hardness results for finitely ambiguous PFA and that emptiness is in {\bf NP} for the class of $k$-ambiguous PFA for every $k > 0$. The emptiness problem for PFA was later shown to also be undecidable even for linearly ambiguous automata in \cite{DJ18}.

\subsection{Our Contributions}

In this paper, we show that the strict and nonstrict emptiness problems are undecidable even for polynomially ambiguous commutative PFA when all matrices are rational. This implies that undecidability holds even when the input words come from a letter monotonic language (since the order of input words is irrelevant, only the number of occurrences  of each letter is important). This combination of restrictions on the PFA significantly increases the difficulty of proving undecidability. The study of PFA over letter monotonic languages is a particularly interesting intermediate model, lying somewhere between single letter alphabets (equivalent to Skolem's problem \cite{HY16}) and PFA defined with multi-letter alphabets, for which most decision problems are undecidable. We also show that the problem remains undecidable even for binary input alphabets, although we only obtain the result for noncommutative PFA and when the input words are from bounded, rather than letter monotonic, languages.

\begin{theorem}\label{mainThm}
The emptiness problem for polynomially ambiguous commutative probabilistic finite automata (and thus when inputs are restricted to letter monotonic languages) is undecidable for strict/non-strict cut-points. The problem remains undecidable for a binary alphabet if letter monotonic languages are replaced by bounded languages and we remove the commutativity restriction on the PFA.
\end{theorem}

We note a few difficulties with proving this result. Firstly, Post's correspondence problem, whose variants are often used for showing undecidability results in such settings, is actually decidable over letter monotonic languages \cite{Halava2009}\footnote{Although it is undecidable in general (i.e. not over a letter monotonic language) with an alphabet with at least five letters \cite{Neary2015}.}. Secondly, although other reductions of undecidable computational problems to matrices are possible, the standard technique of Turakainen (shown in \cite{Tu69}) to modify such matrices to stochastic matrices introduces exponential ambiguity (indeed all such matrices are strictly positive, and thus we might think of such matrices as being \emph{maximally exponentially ambiguous})\footnote{This is due to an essential step of the Turakainen procedure that adds a positive constant offset to each element of every generator matrix, thus making all matrices strictly positive \cite{Tu69}.}. Finally, we note that matrix problems for commutative matrices are often decidable; indeed there are polynomial time algorithms for solving the orbit problem \cite{COW16, KL86} and the vector reachability problem for commutative matrices \cite{BBC94}. Since the matrices commute, it is the Parikh vector of letters of the input word which is important.

We use a reduction of Hilbert's tenth problem and various new encoding techniques to avoid the use of Turakainen's method for converting from weighted to probabilistic automata, so as to retain polynomial ambiguity. We use some techniques to move from non-strict to strict emptiness and to consider binary input alphabets. We then move on to the freeness/injectivity problem to show the following two results. 

\begin{theorem}\label{freenessthm}
The injectivity problem for linearly ambiguous four state probabilistic finite automata is undecidable.
\end{theorem} 

\begin{theorem}\label{nphardthm}
The injectivity problem for linearly ambiguous three-state probabilistic finite automata over letter monotonic languages is NP-hard. 
\end{theorem}

These results are proven via an encoding of the mixed modification PCP and our new encoding technique and the injectivity problem for three state PFA over letter monotonic languages is {\bf NP}-hard via an encoding of a variant of the subset sum problem and a novel encoding technique. We conclude with some open problems.

\section{Preliminaries}\label{sec-not}

\subsection{Linear Algebra}\label{linalg-sec}

Given $A = (a_{ij}) \in\F^{m\times m}$ and $B\in\F^{n\times n},$ we define the direct sum $A\oplus B$ and Kronecker product $A \otimes B$ of $A$ and $B$ by:
\[
A\oplus B=
\left[\begin{array}{@{}c|l@{}}
A & \0_{m,n}\\
\hline
\0_{n,m} & B
\end{array}\right], \quad
A\otimes B=
\left[\begin{array}{cccc}
a_{11}B & a_{12}B & \cdots & a_{1m}B \\ 
a_{21}B & a_{22}B & \cdots & a_{2m}B \\ 
\vdots & \vdots & &\vdots \\ 
a_{m1}B & a_{m2}B & \cdots & a_{mm}B \\ 
\end{array}\right],
\]
where $\0_{i,j}$ denotes the zero matrix of dimension $i \times j$. Note that neither $\oplus$ nor $\otimes$ are commutative in general. Given a finite set of matrices $\mathcal{G} = \{G_1, G_2, \ldots, G_m\} \subseteq \mathbb{F}^{n \times n}$, $\langle\mathcal{G}\rangle$ denotes the semigroup generated by $\mathcal{G}$. We will use the following notations:
$$
\bigoplus_{j=1}^{m} G_j = G_1 \oplus G_2 \oplus \cdots \oplus G_m, \qquad
\bigotimes_{j=1}^{m} G_j = G_1 \otimes G_2 \otimes \cdots \otimes G_m
$$

Given a matrix $G \in \mathbb{F}^{n \times n}$, we inductively define $G^{\otimes k} = G \otimes G^{\otimes (k-1)} \in \mathbb{F}^{n^k \times n^k}$  for $k > 0$ with $G^{\otimes 0} = 1$ as the $k$-fold Kronecker power of $G$. Similarly,  $G^{\oplus k} = G \oplus G^{\oplus (k-1)} \in \mathbb{F}^{n^k \times n^k}$ for $k > 0$ with $G^{\oplus 0}$ being a zero dimensional matrix. The rationale for the base cases is that $G \otimes G^{\otimes 0} = G \otimes 1 = G$ and that $G \oplus G^{\oplus 0} = G$ as expected.

The following properties of $\oplus$ and $\otimes$ are well known and will all be useful later. 

\begin{lemma}\label{kronprop}
Let $A, B, C, D \in \mathbb{F}^{n \times n}$. We note that:
\begin{itemize}
\item Associativity - $(A \otimes B) \otimes C = A \otimes (B \otimes C)$ and $(A \oplus B) \oplus C = A \oplus (B \oplus C)$, thus $A \otimes B \otimes C$ and $A \oplus B \oplus C$ are unambiguous.
\item Mixed product properties: $(A \otimes B)(C \otimes D) = (AC \otimes BD)$ and $(A \oplus B)(C \oplus D) = (AC \oplus BD)$.
\item If $A$ and $B$ are stochastic matrices, then so are $A \oplus B$ and $A \otimes B$. 
\item If $A, B \in \mathbb{F}^{n \times n}$ are both upper-triangular then  so are $A \oplus B$ and $A \otimes B$.
\end{itemize}
\end{lemma}

See \cite{HJ91} for proofs of the first three properties of Lemma~\ref{kronprop}. The fourth property follows directly from the definition of the Kronecker sum and product and is not difficult to prove.

\subsection{Probabilistic Finite Automata (PFA)}\label{pfa-def-sec}
A Probabilistic Finite Automaton (PFA) $\mathcal{A}$ with $n$ states over an alphabet $\Sigma$ is defined as $\mathcal{A}=({\bf x}, \{M_a|a \in \Sigma\}, {\bf y})$ where ${\bf x} \in \mathbb{R}^n$ is the initial probability distribution; ${\bf y} \in \{0, 1\}^n$ is the final state vector and each $M_a \in \mathbb{R}^{n \times n}$ is a (row) stochastic matrix. For a word $w = w_1w_2\cdots w_k \in \Sigma^*$, we define the acceptance probability $f_{\mathcal{A}}: \Sigma^* \to \mathbb{R}$ of $\mathcal{A}$ as:
$$
f_{\mathcal{A}}(w) = {\bf x}^TM_{w_1}M_{w_{2}} \cdots M_{w_k} {\bf y},
$$
which denotes the acceptance probability of $w$.\footnote{Some authors interchange the order of ${\bf x}$ and ${\bf y}$ and use column stochastic matrices, although the two definitions are trivially isomorphic.} If all transition matrices $\{M_a|a \in \Sigma\}$ commute, the the PFA is called a \emph{commutative} PFA.

For any $\lambda \in [0, 1]$ and PFA $\mathcal{A}$ over alphabet $\Sigma$, we define a cut-point language to be: $L_{\geq \lambda}(\mathcal{A}) = \{w \in \Sigma^* | f_{\mathcal{A}}(w) \geq \lambda\}$, and a strict cut-point language $L_{> \lambda}(\mathcal{A})$ by replacing $\geq$ with $>$. The (strict) emptiness problem for a cut-point language is to determine if $L_{\geq \lambda}(A) = \emptyset$ (resp. $L_{> \lambda}(A) = \emptyset$). 

Let $\alp{\ell} = \{x_1, x_2, \ldots, x_\ell\}$ be an $\ell$-letter alphabet for some $\ell>0$. A language $\mathcal{L} \subseteq \Sigma_\ell^*$ is called a \emph{bounded language} if and only if there exist words $w_1, w_2, \ldots, w_m \in \Sigma_\ell^+$ such that $\mathcal{L} \subseteq w_1^* w_2^* \cdots w_m^*$. A language $\mathcal{L}$ is called \emph{letter monotonic} if there exists letters $u_1, u_2, \ldots, u_{m} \in \alp{\ell}$ such that $\mathcal{L} \subseteq u_1^*u_2^*\cdots u_m^*$. One thus sees that letter monotonic languages are more restricted than bounded languages. We will be interested in PFA which are defined over a bounded language or a letter monotonic language $\mathcal{L}$, whereby all input words necessarily come from $\mathcal{L}$. In this case a cut-point language for a PFA $\mathcal{A}$ over bounded/letter monotonic language $\mathcal{L}$ and a probability $\lambda \in [0, 1]$ is defined as $L_{\geq \lambda, \mathcal{L}}(\mathcal{A}) = \{w \in \mathcal{L} | f_{\mathcal{A}}(w) \geq \lambda\}$; similarly for nonstrict cut point languages. We may then ask similar emptiness questions for such languages, as before.

We also study the \emph{freeness/injectivity problem} for PFA. Given a PFA $\mathcal{A}$ over alphabet $\Sigma$, determine whether the acceptance function $f_{\mathcal{A}}(w)$ is injective (i.e. do there exist two distinct words with the same acceptance probability). Such problems can readily be studied when the input words are necessarily derived from a bounded or letter monotonic language. 

\subsection{PFA Ambiguity}\label{amb-sec}

The degree of ambiguity of a finite automaton is a structural parameter, roughly indicating the number of accepting runs for a given input word \cite{WS91}. We here define only those notions required for our later proofs, see \cite{WS91} for full details of these notions and a thorough discussion. 

Let $w \in \Sigma^*$ be an input word of an NFA $\mathcal{N} = (Q, \Sigma, \delta, Q_I, Q_F)$, with $Q$ the set of states, $\Sigma$ the input alphabet, $\delta \subset Q \times \Sigma \times Q$ the transition function, $Q_I$ the set of initial states and $Q_F$ the set of final states. For each $(p, w, q) \in Q \times \Sigma^* \times Q$, let $\textrm{da}_{\mathcal{N}}(p, w, q)$ be defined as the number of all paths for $w$ in $\mathcal{N}$ leading from state $p$ to state $q$. The degree of ambiguity of $w$ in $\mathcal{N}$, denoted $\textrm{da}_{\mathcal{N}}(w)$,  is defined as the number of all \emph{accepting paths} for $w$. The degree of ambiguity of $\mathcal{N}$, denoted $\textrm{da}(\mathcal{N})$ is the supremum of the set $\{\textrm{da}_{\mathcal{N}}(w) | w \in \Sigma^*\}$. $\mathcal{N}$ is called infinitely ambiguous if $\textrm{da}(\mathcal{N}) = \infty$, finitely ambiguous if $\textrm{da}(\mathcal{N}) < \infty$, and unambiguous if $\textrm{da}(\mathcal{N}) \leq 1$. The degree of growth of the ambiguity of $\mathcal{N}$, denoted $\textrm{deg}(\mathcal{N})$ is defined as the minimum degree of a univariate polynomial $h$ with positive integral coefficients such that for all $w \in \Sigma^*$, $\textrm{da}_{\mathcal{N}}(w) \leq h(|w|)$ if such a polynomial exists, or infinity otherwise. 

The above notions relate to NFA. We may derive an analogous notion of ambiguity for PFA by considering an embedding of a PFA $\pfa$ to an NFA $\mathcal{N}$ with the property that for each letter $a \in \Sigma$, if the probability of transitioning from a state $i$ to state $j$ is nonzero under $\pfa$, then there is an edge from state $i$ to $j$ under $\mathcal{N}$ for letter $a$. The degree of (growth of) ambiguity of $\pfa$ is then defined as the degree of (growth of) ambiguity of $\mathcal{N}$. 

We may use the following notions to determine the degree of ambiguity of a given NFA (and thus a PFA by the embedding discussed above) $\mathcal{A}$ as is shown in the theorem which follows. A state $q \in Q$ is called \emph{useful} if there exists an accepting path which visits $q$. See Figure~\ref{ambfigs} for examples.

\noindent {\bf EDA} - There is a useful state $q \in Q$ such that, for some word $v \in \Sigma^*$, $da_{\mathcal{A}}(q, v, q) \geq 2$.

\noindent {\bf $\text{IDA}_d$} - There are useful states $r_1, s_1, \ldots, r_d, s_d \in Q$ and words $v_1, u_2, v_2, \ldots, u_d, v_d \in \Sigma^*$ such that for all $1 \leq \lambda \leq d$, $r_\lambda$ and $s_\lambda$ are distinct and $(r_\lambda, v_\lambda, r_\lambda), (r_\lambda, v_\lambda, s_\lambda), (s_\lambda, v_\lambda, s_\lambda) \in \delta$ and for all $2 \leq \lambda \leq d$, $(s_{\lambda-1}, u_{\lambda}, r_{\lambda}) \in \delta$.

\begin{theorem}[\cite{IR86, Re77, WS91}]\label{crtitthm}
An NFA (or PFA) $\mathcal{A}$ having the EDA property is equivalent to it being exponentially ambiguous. For any $d \in \mathbb{N}$, an NFA (or PFA) $\mathcal{A}$ having property $\text{IDA}_d$ is equivalent to $\textrm{deg}(\mathcal{A}) \geq d$.
\end{theorem}

Clearly, if $\mathcal{N}$ agrees with $\text{IDA}_d$ for some $d > 0$, then it also agrees with $\text{IDA}_1, \ldots, \text{IDA}_{d-1}$. One must be careful with these notions of ambiguity when considering NFA/PFA $\mathcal{A}$, where inputs are restricted to a bounded language $\mathcal{L}$. In such cases, the above criteria do not suffice to determine the ambiguity of  $\mathcal{A}$, since the number of paths must be determined not over $\Sigma^*$, but over words from $\mathcal{L}$. Of course, the degree of ambiguity of $\mathcal{A}$ cannot \emph{increase} by restricting to a bounded input language, but it may decrease. 

As an example, if an NFA $\mathcal{N}$ has property EDA, then there exist three words $w_1, w_2$ and $w_3$, as well as a useful state $q$ such that $w_1w_2w_3$ is an accepting word and $da_{\mathcal{N}}(q, w_2, q) \geq 2$, thus $w_1w_2w_3$ has at least two distinct accepting runs. However, this implies that $da_{\mathcal{N}}(w_1w_2^kw_3) \geq 2^k$ and thus $w_1w_2^kw_3$ has at least $2^k$ accepting runs. Now, if we are given some bounded language $\mathcal{L}$ such that $w_1w_2w_3 \in \mathcal{L}$ and $da_{\mathcal{N}}(q, w_2, q) \geq 2$ then the same implication is not possible, unless $w_2 \in \Sigma$ is a single letter, otherwise there is no guarantee that $w_1w_2^kw_3 \in \mathcal{L}$. Nevertheless, in the results of this paper we will use the standard definitions of ambiguity since the distinction is not relevant in our results as will become clear (and especially in Theorem~\ref{mainThm} for the results on commutative PFA).

\begin{figure}[h]
\centering
\begin{minipage}{0.4\textwidth}
\centering
\begin{tikzpicture}[shorten >=1pt,node distance=2cm,on grid,auto] 
   \node[state] (q_1)   {$q_1$}; 
   \node[state, initial] (q_0) [above right=of q_1] {$q_0$}; 
   \node[state, accepting](q_2) [below right=of q_0] {$q_2$};
    \path[->] 
    (q_0) 		edge [above left] node {$0:\frac{1}{2}$} (q_1)
    			edge  node {$1:\frac{1}{2}$} (q_2)
			edge [loop above] node {$\{0,1\}:\frac{1}{2}$} ()
    (q_1)  edge [loop below] node {$\{0,1\}:\frac{1}{3}$} ()
edge [below] node {$\{0,1\}:\frac{2}{3}$} (q_2)
           (q_2) edge [loop below] node {$\{0,1\}:1$} ();
\end{tikzpicture}
\end{minipage}
\begin{minipage}{0.4\textwidth}
\centering
\begin{tikzpicture}[shorten >=1pt,node distance=2cm,on grid,auto] 
   \node[state, initial] (q_0) {$q_0$}; 
   \node[state, accepting](q_1) [right=of q_0] {$q_1$};
    \path[->] 
    (q_0) 		edge [bend left] node {$a:\frac{1}{2}$} (q_1)
			edge [loop above] node {$a:\frac{1}{2}$} ()
    (q_1)  edge [bend left] node {$a:1$} (q_0);
\end{tikzpicture}
\end{minipage}
\caption{The binary PFA on the left has polynomial (quadratic) ambiguity since it does not satisfy condition EDA. Its transition matrices are upper-triangular; no transition leads from $q_j$ to $q_i$ with $i< j$. The unary PFA on the right satisfies EDA and thus it has exponential ambiguity.}  \label{ambfigs}
\end{figure}
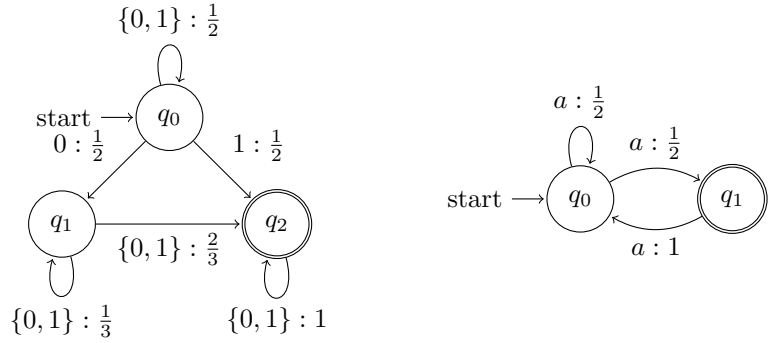

We note the following trivial lemma, which will be useful later.

\begin{lemma}\label{polyamb}
Probabilistic finite automata defined over upper-triangular matrices are polynomially ambiguous.
\end{lemma}
\begin{proof}
Immediate from Theorem~\ref{crtitthm} and property (EDA), since a PFA defined over upper-triangular matrices clearly does not have property (EDA). This is since a transition matrix (for a letter `$a$') which is upper-triangular only defines transitions of the form $\delta(i, a) = j$ where $i \leq j$ and thus the states visited for any run are monotonically nondecreasing.
\end{proof}

\subsection{Reducible Undecidable Problems}\label{mmpcp-sec}

We will require the following undecidable problems for proving later results. The first is a variant of the famous \emph{Post's Correspondence Problem} (PCP).

\begin{prob}[Mixed Modification PCP (MMPCP)]\label{mmpcpdef}
Given a binary alphabet $\Sigma_2$, a finite set of letters $\Sigma = \{s_1, s_2, \ldots, s_{\ell}\}$,  and a pair of homomorphisms $h,g:\Sigma^*\to\Sigma_2^*,$ the MMPCP asks to decide whether there exists a word $w=x_1\dots x_k\in\Sigma^+, x_i\in\Sigma$ such that:
\[  
h_1(x_1)h_2(x_2)\dots h_k(x_k)=g_1(x_1)g_2(x_2)\dots g_k(x_k),
\]
where $h_i,g_i\in\{h,g\},$ and there exists at least one $j$ such that $h_j\neq g_j.$
\end{prob}

\begin{theorem}\hspace{-0.1cm}\cite{CKH96}\label{mmpcp9} - 
The Mixed Modification PCP is undecidable for $|\Sigma| \geq 9$.
\end{theorem}

A second useful undecidable problem is \emph{Hilbert's tenth problem}: does there exist an algorithm to determine if, for  an arbitrary integer polynomial $P(n_1, n_2, \ldots, n_k)$ with $k$ variables, there exist $x_1, x_2, \ldots, x_k \in \mathbb{Z}$ such that: $P(x_1, x_2, \ldots, x_k) = 0$? It is well known that this may be reduced to a problem in
formal power series. It was shown in \cite[p.73]{SS} that the above problem can be reduced to that of determining for a $\mathbb{Z}$-rational formal power series $S \in \zfps$, whether there exists any word $w \in A^*$ such that $(S,w) = 0$. The undecidability of this problem was shown in 1970 by Y.~Matiyasevich (building upon work of Davis, Putman, Robinson and others). For more details, see the excellent reference \cite{Mat}. We may, without loss of generality, restrict the variables to be natural numbers \cite[p.6]{Mat}.

\section{Cut-point languages for polynomially ambiguous commutative PFA}\label{boundedlangsec}

It was proven in \cite{BHH13} that the emptiness problem is undecidable for probabilistic finite automata even when input words are given over a letter monotonic language, i.e., given a PFA $\mathcal{P}$, a cutpoint $\lambda \in [0, 1]$ and a letter monotonic language $\mathcal{L}$, it is undecidable to determine if $\{w \in \mathcal{L} | f_{\pfa}(w) \Delta \lambda\}$ is empty for $\Delta \in \{\leq, <, >, \geq\}$. The constructed PFA of \cite{BHH13} has exponential ambiguity, due to the well-known Turakainen conversion of arbitrary integer matrices into stochastic matrices \cite{Tu69}. Here, we show that the emptiness problem for PFA over letter monotonic languages can also be achieved even when all matrices have polynomial ambiguity by a modified Turakainen procedure. In fact we show that the emptiness problem for PFA with commuting transition matrices is undecidable, and thus only the number, rather than the order, of the input letters matter (i.e. the input word's Parikh vector).

 The following property of the Kronecker product will also be required for the proof of Theorem~\ref{mainThm}.

\begin{lemma}\label{kronprop2}
Let $A_1, \ldots, A_\ell \in \mathbb{F}^{n \times n}$. For any index sequence $(i_1, j_1), \ldots, (i_\ell, j_\ell) \in [1, n] \times [1, n]$, there exists $1 \leq i, j \leq n^\ell$ such that: 
$$
\prod_{m = 1}^{\ell}(A_m)_{i_m, j_m} = \left(\bigotimes_{m= 1}^{\ell} A_m\right)_{i,j}
$$
\end{lemma}
\begin{proof}
The proof proceeds by induction. For the base case when $\ell=1$, we just set $(i, j) = (i_1, j_1)$ and we are done. Assume that the result holds for some $\ell-1$, then for sequence $(i_1, j_1), (i_2, j_2), \ldots, (i_{\ell-1}, j_{\ell-1})$ there exists $1 \leq i', j' \leq n^{\ell-1}$ such that:
$$
\prod_{m = 1}^{\ell-1}(A_m)_{i_m, j_m} = \left(\bigotimes_{m= 1}^{\ell-1} A_m\right)_{i',j'}
$$

By the definition of Kronecker product: 
$$
\left(\left(\bigotimes_{m= 1}^{\ell-1} A_m\right) \otimes A_{\ell}\right)_{ni'+i_\ell,nj'+j_\ell} = \prod_{m = 1}^{\ell-1}(A_m)_{i_m, j_m} \times \left(A_{\ell}\right)_{i_\ell, j_\ell}
$$
as required.
\end{proof}
Note that we can of course work out the particular value of $i$ and $j$, but in general the formula for $i, j$ does not have a nice form when $\ell>2$, and anyway will not be necessary for us, so we settle for an existential proof of such $i$ and $j$ (which can be easily computed if necessary). 

\subsection{Proof of Theorem~\ref{mainThm}}

\begin{proof} We begin with a proof sketch. We use a reduction of Hilbert's tenth problem to show our undecidability result. We first modify the Diophantine equation $P(x_1, \ldots, x_t) = 0$ to $P^h(x_0, x_1, \ldots, x_t) = 0$ such that $P^h$ is nonnegative and homogeneous (each term having the same degree), which is required for later technical reasons. We then denote $P^h$ as a sum of $r$ terms $P^h(x_0, \ldots, x_t) = \sum_{j=1}^{r} T_j(x_0, \ldots, x_t)$. For each term $T_j$, we define a set of $t+1$ integer matrices, corresponding to a $t+1$-letter weighted finite automaton\footnote{A weighted finite automaton (WFA) behaves similarly to an NFA, except edges carry integer weights which are multiplied as edges are traversed and initial and final weight functions, that are characterised as rational formal power series  \cite{Sch61}.} defined by $(u'_j, \{X_{j,\ell} | 0 \leq \ell \leq t\}, v_j)$ such that $(u'_j)^T X_{j,0}^{x_0} X_{j,1}^{x_1} \cdots X_{j,t}^{x_{t}} v'_j = T_{j}(x_0, x_1, \ldots, x_t)$. We show how to convert each such weighted automata  into a polynomially ambiguous probabilistic automata with commuting transition matrices. We then show how to combine these PFA into a larger PFA which encapsulates the sum of terms, and thus the polynomial $P^h$ and define a suitable cutpoint $\lambda$ and letter monotonic language $\mathcal{L}$ such that the non-strict emptiness problem for this PFA is undecidable. We give a technique to obtain the result for strict emptiness and then conclude by considering a binary alphabet and bounded languages. 

\noindent {\bf Encoding Hilbert's tenth problem to weighted finite automata} -  We begin by encoding an instance of Hilbert's tenth problem into a set of integer matrices. Let $P(x_1, x_2, \ldots, x_t) = 0$ be a Diophantine equation. Homogenenization of polynomials is a well known technique, as is used for example in the study of Gr\"{o}bner bases \cite{BC02}, which allows us to convert such a Diophantine equation to $P^h(x_0, x_1, x_2, \ldots, x_t) = 0$ with a new dummy variable $x_0$ such that $P^h$ is a homogeneous polynomial (each term having the same degree $d$) and for which $P^h(x_0, x_1, \ldots, x_t) = P(x_1, x_2, \ldots, x_t)$ when $x_0 = 1$. We thus assume a homogeneous Diophantine equation $P^h(x_0, x_1, \ldots, x_t) = 0$ with implied constraint $x_0 = 1$ which will be dealt with later. Furthermore, we assume that $P^h$ gives nonnegative values, which may be assumed by redefining $P^h = (P^h)^2$, which clearly does not affect whether a zero exists for such a polynomial.

Notice that given $A = \begin{pmatrix} 1 & 1 \\ 0 & 1\end{pmatrix}$, then $A^k = \begin{pmatrix} 1 & k \\ 0 & 1\end{pmatrix}$. We will generalise this property to a set of $t+1$ matrices $A_0, A_1, \ldots, A_t \in \Z^{(t+3) \times (t+3)}$ so that given any tuple $(x_0, x_1, x_2, \ldots, x_t)$, then $x_i$ appears as an element on the superdiagonal of $A_0^{x_0}A_1^{x_1}\cdots A_t^{x_t}$ for each $0 \leq i \leq t$. We will also have the property that each $A_i$ has the same row sum of $2$ for every row, which will be useful when we later convert to stochastic matrices.

We define each matrix $A_i$ for $0 \leq i \leq t+1$ in the following way:
\begin{eqnarray}
A_i = 
\begin{pmatrix}
1 & \delta_{0,i} & 0 & \cdots & 0 & 0 & 1-\delta_{0,i} \\
0 & 1 & \delta_{1,i} & \cdots & 0 & 0 & 1-\delta_{1,i} \\
0 & 0 & 1 & \cdots & 0 & 0 & 1-\delta_{2,i} \\
\vdots & \vdots & \vdots & \ddots & \vdots & \vdots & \vdots \\
0 & 0 & 0 & \cdots & 1 & \delta_{t,i} & 1-\delta_{t,i} \\
0 & 0 & 0 & \cdots & 0 & 1 & 1 \\
0 & 0 & 0 & \cdots & 0 & 0 & 2
\end{pmatrix} \in \N^{(t+3) \times (t+3)}, \label{bmat}
\end{eqnarray}
where $\delta_{\ell, i} \in \{0, 1\}$ is the Kronecker delta (thus $\delta_{i, i} = 1$ and $\delta_{\ell, i} = 0$ for $\ell \neq i$). We also denote $J = A_{t+1}$, noting that this is the matrix~(\ref{bmat}) when all $\delta_{\ell, i}$ have the value $0$. Notice then that every row sum of $A_i$ and $J$ is $2$. The overall structure of each $A_i$ is retained under matrix powers and it is easy to see that:
\begin{eqnarray}
A_i^k = 
\begin{pmatrix}
1 & k\delta_{0,i} & 0 & \cdots & 0 & 0 & 2^k-k\delta_{0,i}-1 \\
0 & 1 & k\delta_{1,i} & \cdots & 0 & 0 & 2^k-k\delta_{1,i}-1 \\
0 & 0 & 1 & \cdots & 0 & 0 & 2^k-k\delta_{2,i}-1 \\
\vdots & \vdots & \vdots & \ddots & \vdots & \vdots & \vdots \\
0 & 0 & 0 & \cdots & 1 & k\delta_{t,i} & 2^k-k\delta_{t,i}-1 \\
0 & 0 & 0 & \cdots & 0 &  1 & 2^k-1 \\
0 & 0 & 0 & \cdots & 0 & 0 & 2^k
\end{pmatrix} \in \N^{(t+3) \times (t+3)}\label{bmatk}
\end{eqnarray}

All row sums of $A_i^k$ are $2^k$ and exactly one element of the superdiagonal is equal to $k$, with all other elements on the superdiagonal (excluding that on row $t+2$) zero. Taking powers of $A_i$ will allow us to choose any nonnegative value of variable $x_i$. Note that $J^k$ has the same form as the matrix of~(\ref{bmatk}) with all $\delta_{\ell,i} = 0$ and acts as a kind of identity matrix, (in its upperleft block) while retaining the $2^k$ row sum. 

Notice that that for $0 \leq i, j \leq t+1$ with $i+1 \neq j$, then $A_iA_j = A_jA_i$, i.e. these matrices commute (similarly for $J=A_{t+1}$). This follows since in a product $A_iA_j$, the main diagonal is always $1$ excluding the bottom right element (which is always $4$) because the matrices are upper triangular. Since $i+1 \neq j$, if we exclude the right column and bottom row, $(A_i)_{1..t+2, 1..t+2} = I_{i} \oplus A \oplus I_{t-i}$ with $I_\ell$ the $\ell$-dimensional identity matrix and $A \in \N^{2 \times 2}$ defined as before. We also have that $(A_j)_{1..t+2, 1..t+2} = I_{j} \oplus A \oplus I_{t-j}$. In this case $A_i A_j$ and $A_jA_i$ have upper left block $I_i \oplus A \oplus I_{j-i-2} \oplus A \oplus I_{t-j}$. Since the right column preserves row sums then matrices $A_i$ and $A_j$ commute so long as $i + 1 \neq j$. We note that $A_i$ and $A_{i+1}$ do not commute however. Therefore, in order to get commutative matrices, we may instead use matrices $A_0, A_2, \ldots, A_{2t}, A_{2t+2}= J$. This requires an increase of dimension to $\N^{(2t+3) \times (2t+3)}$. We proceed with the proof using non-commuting $A_0, \ldots, A_{t+1}$ for ease of exposition, noting that we will later map $A_\ell \in \N^{(t+3) \times (t+3)}$ to $A_{2\ell} \in \N^{(2t+3) \times (2t+3)}$ to obtain commutativity of these matrices. We now show how to compute terms of  $P^h$.

We may write  $P^h(x_0, x_1, \ldots, x_t) = \sum_{j=1}^{r} T_j(x_0, x_1, \ldots, x_t)$, where $T_j$ denotes the $j$'th term of $P^h$, with $P^h$ having $r$ terms. Since $P^h$ is a homogeneous polynomial, each term has the same degree $d$. We may thus write each term as:
\begin{eqnarray}
T_{j}(x_0, x_1, \ldots, x_t) & = & c_{j} R_{j}(x_0, x_1, \ldots, x_t), \label{termeqn}
\end{eqnarray}
with $c_{j} \in \mathbb{Z}$ and $R_{j}(x_0, x_1, \ldots, x_t) = \prod_{\ell = 0}^{t}x_{\ell}^{r_{j,\ell}}$ with $r_{j,\ell} \geq 0$ and $\sum_{\ell=0}^{t} r_{j,\ell} = d$. For convenience, we define a $d$-dimensional vector $s_j = \bigotimes_{\ell = 0}^{t} \ell^{\otimes r_{j, \ell}} \in [0, t]^d$. For example, if $t=3, d=5$ and $T_{j}(x_0, x_1, x_2, x_3) = 6x_0x_1^2x_3^2$, then $R_{j}(x_0, x_1, x_2, x_3) = x_0^1x_1^2x_2^0x_3^2$ and thus $s_{j} = (0, 1, 1, 3, 3)^T \in [0, 3]^5$. By $s_j[i]$ we denote the $i$'th element of vector $s_j$.

We now define $t+1$ matrices corresponding to term $T_j$:
$$
X_{j,i} = \bigotimes_{\ell = 0}^{i-1}J^{\otimes r_{j,\ell}} \otimes A_i^{\otimes r_{j,i}} \otimes \bigotimes_{\ell = i+1}^{t}J^{\otimes r_{j,\ell}},
$$
where $0 \leq i \leq t$. The dimension of such matrices is $(t+3)^{d} \times (t+3)^{d}$ since each submatrix has dimension $(t+3) \times (t+3)$ and we take the $d$-fold Kronecker product. Similarly, we see that the row sum of each $X_{j, i}$ is $2^{d}$ since the row sum of each $A_i$ and $J$ is $2$ and we take a $d$-fold Kronecker product. Clearly then, by the mixed product property (see Lemma~\ref{kronprop}): 
$$
X_{j,i}^k = \bigotimes_{\ell = 0}^{i-1}(J^k)^{\otimes r_{j,\ell}} \otimes (A_i^k)^{\otimes r_{j,i}} \otimes \bigotimes_{\ell = i+1}^{t}(J^k)^{\otimes r_{j,\ell}},
$$
for any $k \geq 0$. In the example when $r_{j,0} = 1$, $r_{j,1} = 2$, $r_{j,2} = 0$, and $r_{j,3} = 2$, then $X_{j,1} = J^{\otimes 1} \otimes A_1^{\otimes 2} \otimes J^{\otimes 0} \otimes J^{\otimes 2} = J^{\otimes 1} \otimes  A_1^{\otimes 2} \otimes J^{\otimes 2}$. We then see that $X_{j,1}^k =(J^k)^{\otimes 1} \otimes  (A_1^k)^{\otimes 2} \otimes (J^k)^{\otimes 2}$. 

Now, we see that:
\begin{eqnarray}
X_{j,0}^{x_0} X_{j,1}^{x_1} \cdots X_{j,t}^{x_{t}}  & = &  \prod_{i=0}^{t} \left( \bigotimes_{\ell = 0}^{i-1}(J^{x_i})^{\otimes r_{j,\ell}} \otimes (A_i^{x_i})^{\otimes r_{j,i}} \otimes \bigotimes_{\ell = i+1}^{t}(J^{x_i})^{\otimes r_{j,\ell}} \right) \label{kronprod2pre} \\
 & = & \bigotimes_{\ell = 0}^{d} \left(D_{\ell, 0}^{x_0} D_{\ell, 1}^{x_1} \cdots D_{\ell, t}^{x_t}\right), \label{kronprod2}
\end{eqnarray}
where $D_{\ell, i} \in \{J, A_i\}$ for $0 \leq i \leq t$. The derivation of Eqn~(\ref{kronprod2}) from Eqn~(\ref{kronprod2pre}) follows by the mixed product property of the Kronecker product (Lemma~\ref{kronprop}). For each product $D_{\ell, 0}^{x_0} D_{\ell, 1}^{x_1} \cdots D_{\ell, t}^{x_t}$, we see that $D_{\ell, s_j[\ell]} = A_{s_j[\ell]}$ and $D_{\ell, j} = J$ for all $0 \leq j \leq d$ with $j \neq s_j[\ell]$. To continue our running example of $s_{j} = (0, 1, 1, 3, 3)^T \in [0, 3]^5$, we see that:
\begin{eqnarray*}
X_{j,0}^{k_0} & = & A_0^{k_0} \otimes (J^{k_0})^{\otimes 2} \otimes (J^{k_0})^{\otimes 0} \otimes (J^{k_0})^{\otimes 2}  \\
X_{j,1}^{k_1} & = & J^{k_1} \otimes (A^{k_1})^{\otimes 2} \otimes (J^{k_1})^{\otimes 0} \otimes (J^{k_1})^{\otimes 2}  \\
X_{j,2}^{k_2} & = & J^{k_2} \otimes (J^{k_2})^{\otimes 2} \otimes (A^{k_2})^{\otimes 0} \otimes (J^{k_2})^{\otimes 2}  \\
X_{j,3}^{k_3} & = & J^{k_3} \otimes (J^{k_3})^{\otimes 2} \otimes (J^{k_3})^{\otimes 0} \otimes (A^{k_3})^{\otimes 2}  \\
\end{eqnarray*}
Note that in each `column' of the Kronecker product above, we have exactly one $A_i$ matrix, with the other elements $J$ matrices. Then we see that, assuming matrices $\{A_i | 1 \leq i \leq t\}\cup J$ commute (e.g. by using our previous mapping to increase the dimension of each $A_i$ which we now assume), then by the mixed product property of Kronecker products:
\[
X_{j,0}^{k_0}X_{j,1}^{k_1}X_{j,2}^{k_2}X_{j,3}^{k_3} = (A_0^{k_0}J^{k_1+k_2+k_3}) \otimes (A_1^{k_1}J^{k_0+k_2+k_3})^{\otimes 2} \otimes (A_3^{k_3}J^{k_0+k_1+k_2})^{\otimes 2} 
\]

Back to the more general case since now $\{A_i | 1 \leq i \leq t\}\cup J$ commute,  we may thus rewrite (\ref{kronprod2}) as:
\begin{eqnarray}
X_{j,0}^{x_0} X_{j,1}^{x_1} \cdots X_{j,t}^{x_{t}}  = \bigotimes_{\ell = 0}^{d} \left(A_{s_{j}[\ell]}^{x_{s_{j}[\ell]}} J^{\overline{x_{s_{j}[\ell]}}}\right),\, \text{ where } \, \overline{x_{s_{j}[\ell]}} = \sum_{\substack{0 \leq q \leq t \\ q \neq s_j[\ell]}} x_{q}  \label{kronprod2post}
\end{eqnarray}
By Lemma~\ref{kronprop2}, we see that some element of $X_{j,0}^{x_0} X_{j,1}^{x_1} \cdots X_{j,t}^{k_{t}}$ is thus equal to $R_{j}(x_0, x_1, \ldots, x_t)$, since there is an element on the superdiagonal of $A_{s_{j}[\ell]}^{x_{s_{j}[\ell]}}J^{\overline{x_{s_{j}[\ell]}}}$, namely $(A_{s_j[\ell]}^{x_{s_j[\ell]}}J^{\overline{x_{s_{j}[\ell]}}})_{s_j[\ell], s_j[\ell]+1}$, equal to $x_{s_{j}[\ell]}$ for each $0 \leq \ell \leq d$. Let us assume that $R_{j}(x_0, x_1, \ldots, x_t)$ appears at row $i_1$ and column $i_2$. Now, we may define a vector $u'_j =  c_j  e_{i_1}$ and $v_j' = e_{i_2}$ where $c_j$ is the coefficient of term $T_j$ as in Eqn~(\ref{termeqn}) and $e_{i_1}, e_{i_2} \in \Z^{(2t+3)^{d}}$ are standard basis vectors. We may now see that: 
\begin{eqnarray}
(u'_j)^T X_{j,0}^{x_0} X_{j,1}^{x_1} \cdots X_{j,t}^{x_{t}} v'_j = c_{j} R_{j}(x_0, x_1, \ldots, x_t) = T_{j}(x_0, x_1, \ldots, x_t) \label{computeTerm}
\end{eqnarray}

In order to derive the sum of the $r$ such terms $\sum_{j=1}^{r}T_j(x_0, x_1, \ldots, x_t)$, we will utilise the \emph{direct sum}. For $0 \leq \ell \leq t$, we define $Y_{\ell}'$ by:
$$
Y'_\ell = \bigoplus_{j=1}^r X_{j, \ell} \in \N^{r(2t+3)^{d} \times r(2t+3)^{d}} 
$$

Defining $u'' = \oplus_{j=1}^{r} u'_j$ and $v'' = \oplus_{j=1}^{r} v'_j$, we now have a weighted finite automaton $(u'', \{Y'_{\ell} | 0 \leq \ell \leq t \}, v'')$ such that:
\[ P^h(x_0, x_1, \ldots, x_t) = u''^T (Y'_{0})^{x_0} (Y'_{1})^{x_1} \cdots (Y'_{t})^{x_t} v''
\] 
We now work to show how this can be converted to a probabilistic finite automaton, while retaining polynomial ambiguity and the commutativity of all matrices.

\noindent {\bf Encoding to a probabilistic finite automaton} -  
We first modify each $Y'_\ell$ so that they are row stochastic. We recall that the row sum of each $A_\ell$ and $J$ is $2$. Therefore, the row sum of each $X_{j,\ell}$ is $2^d$, since $X_{j,\ell}$ is a $d$-fold Kronecker product of $A_i$ and $J$ matrices. Then the row sum of each $Y_\ell'$ is also $2^d$ since direct sums do not modify the row sum. We thus see that $Y_{\ell} = \frac{1}{2^d}Y'_\ell$ is row stochastic. 

We now consider the coefficients of each term. We previously defined $u_j'$ by $u_j' = c_je_{i_1}$ and we may consider taking the Kronecker sum of each $u_j'$ before normalising the resulting vector (normalising according to $L^1$ norm). We face an issue however, since some coefficients $c_j$ may be negative and thus the resulting vector is not stochastic (it must be nonnegative). Fortunately we may modify a technique utilised by Bertoni \cite{Be74} to solve this issue. Given a PFA for which $u^T X v = \lambda \in [0, 1]$, then by defining $v' = {\bf 1} - v$ where ${\bf 1}$ is the all-one vector of appropriate dimension (i.e. swapping between final and non final states), then $u^T X v' = 1 - \lambda \in [0, 1]$. 

Let us define $u_j = |c_j| e_{i_1}$, which is similar to $u_j'$ defined previously, but using the absolute value of the corresponding coefficient. Now, since each $X_{j, \ell}$ has a row sum of $2^d$ and $u_j$ is of length $|c_j|$ ($L^1$ norm), then Eqn.~(\ref{computeTerm}) can be adapted to the following:
\begin{eqnarray}
(u_j)^T X_{j,0}^{x_0} X_{j,1}^{x_1} \cdots X_{j,t}^{x_{t}} ({\bf 1} - v_j) & = & |c_j| 2^{d(x_0 + x_1 + \ldots + x_t)} - |c_{j}| R_{j}(x_0, x_1, \ldots, x_t) \nonumber \\ & = & |c_j| 2^{d(x_0 + x_1 + \ldots + x_t)} +  T_{j}(x_0, x_1, \ldots, x_t) \label{computeTermNeg}
\end{eqnarray}

Let us assume, without loss of generality, that we have arranged the terms of $P^h$ such that those terms with a positive coefficient (positive terms) appear first, followed by those with a negative coefficient (negative terms). Since we have $r$ terms in $P^h$, there exists some $1 \leq r' \leq r$ such that we have $r'$ postive and $r-r'$ negative terms. 

We define $v = \bigoplus_{j=1}^{r'} v_j  \oplus \bigoplus_{j=r'+1}^{r} ({\bf 1} - v_j) \in \{0, 1\}^{r(2t+3)^{d}}$ as the final vector, so that we take the Kronecker sum of all final vectors, but we swap final and non-final states for the negative terms. 

We now define the initial vector $u$, which must be a probability distribution. 
Let $g = \sum_{j=1}^{r} |c_j|$ be the sum of absolute values of coefficients and define $u = \frac{1}{g}\bigoplus_{j=1}^{r} u_j \in [0, 1]^{r(2t+3)^{d}}$. Note that $u$ is stochastic (a probability distribution).  

We now see that:
\begin{eqnarray}
 & & u^T Y_0 Y_1^{x_1} \cdots Y_{t}^{x_t} v   \label{newForm1} \\
 & =  & \frac{\sum_{j = 1}^{r'} u_{j} \left( \bigotimes_{\ell = 0}^{d} A_{s_{j}[\ell]}^{x_{s_{j}[\ell]}} \otimes J^{\overline{x_{s_{j}[\ell]}}} \right) v_{j} + \sum_{j = r'+1}^{r} u_{j} \left( \bigotimes_{\ell = 0}^{d} A_{s_{j}[\ell]}^{x_{s_{j}[\ell]}} \otimes J^{\overline{x_{s_{j}[\ell]}}}\right) ({\bf 1} - v_j)}{g2^{d(1 + x_1+\cdots + x_t)}}\nonumber
\end{eqnarray}
Here we used the definition of matrices $Y_i$ and Eqn.~(\ref{kronprod2post}) to rewrite the expressions for $X_{j, 0} \cdots X_{j, t}$. Notice that the power of $Y_0$ (i.e. $x_0$) is set at $1$, since that constraint is required by the conversion from a standard Diophantine polynomial to a homogeneous one as explained previously. Now, using Eqn.~(\ref{computeTerm}) and Eqn.~(\ref{computeTermNeg}), we can rewrite Eqn.~(\ref{newForm1}) as:
\begin{eqnarray}
 &  & \frac{\sum_{j = 1}^{r'} T_j(1, x_1, \ldots, x_t) + \sum_{j = r'+1}^{r} \left(|c_j| 2^{d(1 + x_1+\ldots + x_t)} + T_j(1, x_1, \ldots, x_t) \right)}{g2^{d(1 + x_1+\cdots + x_t)}} \\
 & = & \frac{\sum_{j = r'+1}^{r} |c_j|}{g} + \frac{\sum_{j = 1}^{r'} T_j(1, x_1, \ldots, x_t) + \sum_{j = r'}^{r} T_j(1, x_1, \ldots, x_t) }{g2^{d(1 + x_1+\cdots + x_t)}} \\
 & = & \frac{g'}{g} + \frac{P^h(1, x_1, \ldots, x_t)}{g2^{d(1 + x_1+\cdots + x_t)}}, \label{finalForm}
\end{eqnarray} 
where $g' = \sum_{j = r'+1}^{r} |c_j|$. We therefore define $\pfa = (u, \{Y_{a} | a \in \Sigma_t \}, v)$ and $\Sigma_t = \{0, 1, \ldots, t\}$ as our PFA, with letter monotonic language $\mathcal{L} = 01^*2^*\cdots t^*$ and 
$\lambda =  \frac{g'}{g} \in [0, 1] \cap \Q$ as the cut-point. There exists some word $w = 01^{x_1}2^{x_2} \cdots t^{x_t} \in \mathcal{L}$ such that $f_\pfa(w) \leq \lambda$ if and only if $P^h(1, x_1, x_2, \ldots, x_t) = 0$. Therefore the non-strict emptiness problem for $\pfa$ is undecidable on letter monotonic languages. Since $\pfa$ is upper-triangular, then it is polynomially ambiguous.
We note the surprising fact that all generator matrices are in fact \emph{commutative} (each $X_{j,i}$ is commutative and direct sums do not affect commutativity), which leads to the undecidability of non-strict cut-points for polynomially ambiguous PFA defined over commutative matrices. In this case, the order of the input word in irrelevant, only the Parikh vector of alphabet letters is important. To remove the constraint on using letter `0' once, we may redefine $u = uY_0$ and $\mathcal{L} = 1^*2^*\cdots t^*$ to remove $Y_0$ and all constraints on $\mathcal{L}$. The result now holds for commutative PFA as required.

We have shown the undecidability of emptiness of $\{w: f_{\mathcal{P}}(w) \leq \lambda \textrm{ and } w \in \mathcal{L}\}$. It remains to show how to modify the PFA so that we obtain undecidability for inequalities $\geq, <,$ and $>$, and when the alphabet is binary (but then over a bounded language rather than letter monotonic language and for non-commuting matrices).

\noindent {\bf Emptiness for strict cutpoints is undecidable} -  Let us first prove that determining the emptiness of $\{w: f_{\mathcal{P}}(w) < \lambda \textrm{ and } w \in \mathcal{L}\}$ is undecidable; i.e. the strict emptiness problem. We proceed with a technique inspired by \cite{GO10}. Notice that for all $w \in \mathcal{L}$, then $f_{\pfa}(w)$ is of the form: 
\begin{eqnarray}
\frac{g'}{g} + \frac{P^h(1, x_1, \ldots, x_t)}{g2^{d(1 + x_1+\cdots + x_t)}} & = & \lambda + \frac{P^h(1, x_1, \ldots, x_t)}{g2^{d|w|}}, \label{xCalc}
\end{eqnarray}
as can be seen from~(\ref{finalForm}), where $\lambda = \frac{g'}{g} \in \mathbb{Q} \cap [0, 1]$ and $P^h(1, x_1, \ldots, x_t) \in \mathbb{N}$, since $P^h$ is nonnegative and Diophantine. Therefore $f_{\pfa}(w) \leq \lambda$ if and only if $f_{\pfa}(w) < \lambda + \frac{1}{g2^{d|w|}}$. Let us adapt $\pfa$ in the following way. We add three new states, denoted $q_0, q_F$ and $q_*$. State $q_0$ is a new initial state which, for any input letter, has probability $\frac{1}{2r}$ of moving to each of the $r$ initial states of $\pfa$ and probability $\frac{1}{2}$ to move to new state $q_F$. Recall that $\pfa$ has $r$ initial states, one for each term. State $q_F$ is a new final state that remains in $q_F$ for any input letter with probability $1-\frac{1}{g2^{d}}$ and moves to a new non accepting absorbing sink state $q_*$ with probability $\frac{1}{g2^{d}}$. Let us denote the new PFA $\pfa_{<}$. We now see that for any $a \in \Sigma_t$:
\[ f_{\pfa_<}(aw) = \frac{1}{2}f_{\pfa}(w) + \frac{1}{2}\left(1-\frac{1}{g^{|w|}2^{d|w|}}\right)
\]
If there exists some word $w_1 \in \mathcal{L}$ such that $f_{\pfa}(w_1) \leq \lambda$ then $f_{\pfa}(w_1) = \lambda$ and thus:
\[f_{\pfa_{<}}(aw_1) = \frac{1}{2}\lambda + \frac{1}{2}\left(1-\frac{1}{g^{|w_1|}2^{d|w_1|}}\right) < \frac{1}{2}(\lambda+ 1).
\]
For any $w_2 \in \mathcal{L}$ such that $f_{\pfa}(w_2) > \lambda$ then $f_{\pfa}(w_2) \geq \lambda + \frac{1}{g2^{d|w_2|}}$ by~(\ref{xCalc}). Thus: 
\[f_{\pfa_{<}}(aw_2) \geq \frac{1}{2}(\lambda + \frac{1}{g2^{d|w_2|}})+ \frac{1}{2}\left(1-\frac{1}{g^{|w_2|}2^{d|w_2|}}\right) > \frac{1}{2}(\lambda + 1).
\]
Therefore determining if there exists $w \in \mathcal{A}$ such that $f_{\pfa_<}(w) < \frac{1}{2}(\lambda+ 1)$, i.e. the strict emptiness problem for $\pfa_<$ on cutpoint $\frac{1}{2}(1 + \lambda)$ with letter monotonic language $\mathcal{L}$,  is undecidable as required. Note that the modifications to $\pfa$ retain polynomial ambiguity since $q_0$ and $q_F$ have no incoming (non self looping) edges and $q_*$ has no outgoing edges, therefore property EDA does not hold. We may also see that commutativity of the PFA is unaffected since $\mathcal{P}_{<}$ is identical to $\mathcal{P}$ except for adding three new states, each of which behave identically for all input letters.

Finally, let $\pfa_{\geq}$ be a PFA identical to $\pfa$ except that all final states and non-final states are interchanged. Clearly then $f_{\pfa} = 1-f_{\pfa_{\geq}}$ and thus since emptiness of $\{w: f_{\mathcal{P}}(w) \leq \lambda \textrm{ and } w \in \mathcal{L}\}$ is undecidable, we see that emptiness of $\{w: f_{\mathcal{P}_{\geq}}(w) \geq \lambda \textrm{ and } w \in \mathcal{L}\}$ is also undecidable. A similar idea shows undecidability for inequality $>$, mutatis mutandis.

\noindent {\bf Binary alphabets and bounded languages} -  We conclude this section by showing the undecidability of emptiness of polynomially ambiguous PFA over a binary alphabet and bounded languages. To do so, we utilise a modification of a standard trick. Let $\pfa_{\geq} = (u, \{Y_{a} | a \in \Sigma_t \}, v)$ as above, where the dimension of the vectors (and square matrices) is $\varsigma = r(2t+3)^d$. Let $I_k$ denote the $k \times k$ identity matrix for $k> 0$. Define $Y = Y_0 \oplus Y_1 \oplus \cdots \oplus Y_t$ and $Z = \begin{pmatrix} 0 & I_{t\varsigma} \\ I_\varsigma & 0 \end{pmatrix}$ so that $Y, Z \in \mathbb{Q}^{(t+1)\varsigma \times (t+1)\varsigma}$ and let $u' = (u^T, 0, \ldots, 0)^T$ and $v' = (v^T, 0,  \ldots, 0)^T$, with $u', v' \in \mathbb{Q}^{(t+1)\varsigma}$. It is not difficult to verify that $Z^{t+1} =I_{(t+1)\varsigma}$ and:
\[
Z^{i} Y Z^{t+1-i} = Y_i \oplus Y_{i+1} \oplus \cdots \oplus  Y_t \oplus Y_0 \oplus \cdots Y_i, 
\]
where $0 \leq i \leq t$, thus we permute the blocks of $Y$. Any product containing at least one $Y$ factor thus has a top left $\varsigma \times \varsigma$ block of either the zero matrix or some $Y_i$. For any matrix $Y_{i_1} \cdots Y_{i_p} \in \langle Y_0, \ldots, Y_t \rangle$, there exists a matrix in $\langle Y, Z \rangle$ where $Y_{i_1} \cdots Y_{i_p}$ appears as the top left block, specifically:
\[
Z^{i_1} Y Z^{t+1-i_1} \cdot Z^{i_2} Y Z^{t+1-i_2} \cdots Z^{i_p} Y Z^{t+1-i_p}
\]
Since only the first $\varsigma$ elements of $u'$ and $v'$ are nonzero, then:
\[
u'^T Z^{i_1} Y Z^{t+1-i_1} \cdot Z^{i_2} Y Z^{t+1-i_2} \cdots Z^{i_p} Y Z^{t+1-i_p} v' = u^T Y_{i_1} \cdots Y_{i_p} v
\]
If the top left $\varsigma \times \varsigma$ block of some $F \in \langle Y, Z\rangle$ is zero, then clearly $u'^T F v = 0$. Notice that $Y$ and $Z$ are stochastic matrices (though no longer commutative) and remain polynomially ambiguous (since only the product of the top left blocks of $Y, Z$ is important given that $u', v'$ are only nonzero for their first $\varsigma$ elements and the top left blocks are upper triangular), therefore the strict emptiness problem for $\pfa' = (u', \{Y, Z\}, v')$ is undecidable over bounded language $\mathcal{L}' = (z^{0}yz^{t+1})^* (z^{1}yz^{t})^* \cdots (z^{t}yz^{1})^*$ with $y$ mapping to $Y$ and $z$ mapping to $Z$.
\end{proof}

\section{Injectivity problems for polynomially ambiguous PFA}\label{freenesssec}

We now study the injectivity of acceptance probabilities of polynomially ambiguous PFA. The next result begins with an adapted proof technique from \cite{BHH13}, where the undecidability of the injectivity problem (called the freeness problem in \cite{BHH13}, although we here rename it injectivity) was shown for exponentially ambiguous PFA over five states. We show that the injectivity problem remains undecidable even when the PFA is polynomially ambiguous and over four states by using our new encoding technique (avoiding the Turakainen procedure which increases the matrix dimensions by two and generates an exponentially ambiguous PFA).

\subsection{Proof of Theorem~\ref{freenessthm}}

\begin{proof}
Let $\Sigma=\{x_1,x_2,\dots,x_{n-2}\}$ and $\Delta=\{x_{n-1},x_n\}$ be distinct alphabets and $h, g:\Sigma^* \to \Delta^*$ be an instance of the mixed modification PCP. The naming convention will become apparent below. 
We define two injective mappings $\alpha, \beta:(\Sigma \cup \Delta)^*\to\mathbb{Q}$ by: 
\[ 
     \begin{array}{l}
           \alpha(x_{i_1}x_{i_2} \cdots x_{i_m})=\Sigma_{j=1}^m i_j(n+1)^{j-1},\\
           \beta(x_{i_1}x_{i_2} \cdots x_{i_m})=\Sigma_{j=1}^m i_j(n+1)^{-j},  
     \end{array}
\]
where $\alpha(\varepsilon) = \beta(\varepsilon) = 0$ and each $1 \leq i_j \leq n$. Thus $\alpha$ represents $x_{i_1}x_{i_2} \cdots x_{i_m}$ as a reverse $(n+
1)$-adic number and $\beta$ represents $x_{i_1}x_{i_2} \cdots x_{i_m}$ as a fractional number $(0.x_{i_1}x_{i_2} \cdots x_{i_m})_{(n+1)}$ (e.g. if $n=9$, then $x_1x_2x_3$ is represented as $\alpha(x_1x_2x_3) = 321_{10}$ and $\beta(x_1x_2x_3) = 0.123_{10}$, where subscript $10$ denotes base $10$). Note that $\forall w \in (\Sigma \cup \Delta)^*, \alpha(w) \in \mathbb{N}$ and $\beta(w) \in [0,1) \cap \mathbb{Q}$. It is not difficult to see that $\forall w_1, w_2 \in (\Sigma \cup \Delta)^*, (n+1)^{|w_1|}\alpha(w_2) + \alpha(w_1) = \alpha(w_1w_2)$ and $(n+1)^{-|w_1|}\beta(w_2) + \beta(w_1) = \beta(w_1w_2)$.

Define $\gamma'':(\Sigma\cup\Delta)^*\times(\Sigma\cup\Delta)^*\to\mathbb{Q}^{3\times3}$ by:
\[
\gamma''(u,v)=\begin{pmatrix} (n+1)^{|u|} &0  & \alpha(u)\\ 0 & (n+1)^{-|v|}  & \beta(v)\\ 0& 0   & 1 \end{pmatrix}.
\]
It is easy to verify that $\gamma''(u_1,v_1)\gamma''(u_2,v_2)=\gamma''(u_1u_2,v_1v_2),$ i.e., $\gamma''$ is a homomorphism.  

Let $\mathcal{G}''=\{\gamma''(x_i,g(x_i)),\gamma''(x_i,h(x_i))|x_i\in\Sigma,1\leq i\leq n-2\}$, $\mathcal{S}'' = \langle \mathcal{G}'' \rangle$, $\rho''=(1,1,0)^T$ and $\tau''=(0,0,1)^T$. Assume that there exist $M_1 = G_{i_1} G_{i_2} \cdots G_{i_t} \in \langle\mathcal{G}''\rangle$ and $M_2 = G_{j_1} G_{j_2} \cdots G_{j_{t'}} \in \langle\mathcal{G}''\rangle$ such that $t \neq t'$ or else at least one $G_{i_p} \neq G_{j_p}$ where $1 \leq p \leq t$ and $\lambda = \rho''^T M_1 \tau'' = \rho''^T M_2 \tau''$. We see that:
\[
 \begin{array}{l}
\lambda = \rho''^TM_1\tau''= \alpha(x_{i_1}x_{i_2}\cdots x_{i_t})+\beta(f_{1}(x_{i_1})f_{2}(x_{i_2})\cdots f_{t}(x_{i_t})), \\
\lambda = \rho''^TM_2\tau''= \alpha(x_{j_1}x_{j_2}\cdots x_{j_{t'}})+\beta(f'_{1}(x_{j_1})f'_{2}(x_{j_2})\cdots f'_{t'}(x_{j_{t'}})),
\end{array}
\]
where each $f_{i}, f'_i \in\{g,h\}$. Since $\alpha(w) \in \mathbb{N}$ and $\beta(w) \in (0,1) \cap \mathbb{Q}$, $\forall w \in (\Sigma\cup\Delta)^*$, injectivity of $\alpha$ and $\beta$ implies that if $\rho''^T M_1 \tau'' = \rho''^T M_2 \tau''$, then $t=t'$ and $i_k = j_k$ for $1 \leq k \leq t$. Furthermore, if $\rho^T M_1 \tau = \rho^T M_2 \tau$, we have that $\beta(f_{1}(x_{i_1})f_{2}(x_{i_2})\cdots f_{t}(x_{i_t})) = \beta(f'_{1}(x_{i_1})f'_{2}(x_{i_2})\cdots f'_{t}(x_{i_t}))$ and since at least one $f_p \neq f'_p$ for $1 \leq p \leq t$ by our above assumption, then this corresponds to a correct solution to the MMPCP instance $(h, g)$. On the other hand, if there does not exist a solution to $(h, g)$, then $\beta(f_{1}(x_{i_1})f_{2}(x_{i_2})\cdots f_{t}(x_{i_t})) \neq \beta(f'_{1}(x_{i_1})f'_{2}(x_{i_2})\cdots f'_{t}(x_{i_t}))$, and injectivity of $\beta$ implies that $\rho''^T M_1 \tau'' \neq \rho''^T M_2 \tau''$. 

We now use our new technique to encode such matrices and vectors to a linearly ambiguous four state PFA. We first define a mapping $\gamma':(\Sigma\cup\Delta)^*\times(\Sigma\cup\Delta)^*\to\mathbb{N}^{3\times3}$ to make all matrices be nonnegative integral:
\begin{eqnarray*}
\gamma'(u,v) & = & (n+1)^{|v|} \gamma''(u, v)  = \begin{pmatrix} (n+1)^{|u|+|v|} &0  & (n+1)^{|v|}\alpha(u)\\ 0 & 1  & (n+1)^{|v|}\beta(v)\\ 0& 0   & (n+1)^{|v|} \end{pmatrix} \in \mathbb{N}^{3\times3}
\end{eqnarray*}
We next define the following morphism $\gamma:(\Sigma\cup\Delta)^*\times(\Sigma\cup\Delta)^*\to\mathbb{Q}^{4\times 4}$ to make all such matrices be row stochastic:
\[
\gamma(u,v)= (n+1)^{-k}  \begin{pmatrix} (n+1)^{|u|+|v|} &0  & (n+1)^{|v|}\alpha(u) & \delta_1 \\ 0 & 1  & (n+1)^{|v|}\beta(v) & \delta_2 \\ 0& 0   & (n+1)^{|v|} & \delta_3 \\ 0 & 0 & 0 & \delta_4 \end{pmatrix},
\]
where $\delta_j \in \mathbb{N}$ are chosen so that the row sum of each row of $\gamma(u, v)$ is $(n+1)^k$ for some $k$. Any sufficiently large $k$ can be used so long as each row has the same sum $(n+1)^k$ and thus $\gamma(u,v)$ becomes row stochastic. We use the same $k$ value for all matrices of $\mathcal{G}$ which we define as $\mathcal{G}=\{\gamma(x_i,g(x_i)),\gamma(x_i,h(x_i))|x_i\in\Sigma,1\leq i\leq n-2\}$, so that $\mathcal{S} = \langle \mathcal{G} \rangle$, and finally $\rho=(1,1,0, 0)^T$ and $\tau=(0,0,1, 0)^T$ are the initial and final state vectors respectively.

Assume that there exist $M_1 = G_{i_1} \cdots G_{i_t} \in \langle\mathcal{G}\rangle$ and $M_2 = G_{j_1} \cdots G_{j_{t'}} \in \langle\mathcal{G}\rangle$ such that $t \neq t'$ or else at least one $G_{i_p} \neq G_{j_p}$ for $1 \leq p \leq t$ and $\lambda = \rho^T M_1 \tau = \rho^T M_2 \tau$. We see that:
\[
 \begin{array}{l}
\lambda = \rho^TM_1\tau= (n+1)^{-kt} \left(\alpha(x_{i_1}x_{i_2}\cdots x_{i_t})+\beta(f_{1}(x_{i_1})f_{2}(x_{i_2})\cdots f_{t}(x_{i_t}))\right), \\
\lambda = \rho^TM_2\tau= (n+1)^{-kt'} \left(\alpha(x_{j_1}x_{j_2}\cdots x_{j_{t'}})+\beta(f'_{1}(x_{j_1})f'_{2}(x_{j_2})\cdots f'_{t'}(x_{j_{t'}}))\right),
\end{array}
\]
where each $f_{i}, f'_i \in\{g,h\}$. If $t=t'$, then the same argument as previously shows that $i_k = j_k$ for $1 \leq k \leq t$. If $t \neq t'$, assume without loss of generality that $t' < t$. In this case we see that:
\[
(n+1)^{-kt''} \left(\alpha(x_{i_1}\cdots x_{i_t})+\beta(f_{1}(x_{i_1})\cdots f_{t}(x_{i_t}))\right) = \alpha(x_{j_1}\cdots x_{j_{t'}})+\beta(f'_{1}(x_{j_1})\cdots f'_{t'}(x_{j_{t'}})),
\]
where $t'' = t-t'$. This is a contradiction however since the number of nonzero digits (where a digit is understood base $(n+1)$ here) in the left hand side of this expression is exactly $2t$, and the number of digits in the right expression is $2t' < 2t$. Note that the multiplication by $(n+1)^{-kt''}$ does not alter the number of nonzero digits, it is only a right shift of all digits, $kt''$ times. Thus, since the left and right sides have a different number of nonzero digits they cannot be equal and thus $t=t'$ as required.
\end{proof}

\subsection{Proof of Theorem~\ref{nphardthm}}

\begin{proof}
We use a reduction from the equal subset sum problem, defined thus: given a set of positive integers $S = \{x_1, x_2, \ldots, x_k\} \subseteq \N$, do there exist two disjoint nonempty subsets $S_1, S_2 \subseteq S$ such that $\sum_{\ell \in S_1} \ell = \sum_{m \in S_2} m$? This problem is known to be NP-complete \cite{WY92}. Note that although there is a requirement that the sets $S_1$ and $S_2$ be disjoint, this is not crucial so long as $S_1 \neq S_2$ (since if some element $x_j$ is in both $S_1, S_2$, then the equality also holds when $x_j$ is removed from both sets). We may therefore require that $S_1 \neq S_2$, with both nonempty such that the sum of elements of each set is identical. We define the set of matrices $M = \{A_i, B_i | 1 \leq i \leq k\} \subseteq \Q^{3 \times 3}$ in the following way:
$$
A_i = \frac{1}{x_i+1} \begin{pmatrix} 1 & x_i & 0 \\ 0 & 1 & x_i \\ 0 & 0 & x_i + 1 \end{pmatrix}, \quad B_i = \frac{1}{x_i+1} \begin{pmatrix} 1 & 0 & x_i \\ 0 & 1 & x_i \\ 0 & 0 & x_i + 1 \end{pmatrix}
$$
Note that $A_i$ and $B_i$ are thus row stochastic. Let $u = (1, 0, 0)^T$ be the initial probability distribution, $v = (0,1,0)^T$ be the final state vector and let $\pfa = (u, \{A_i, B_i\}, v)$ be our PFA. Define letter monotonic language $\mathcal{L} = (a_1 | b_1) (a_2 | b_2) \cdots (a_k | b_k) \subseteq a_1^*b_1^*a_2^*b_2^*  \cdots a_k^*b_k^*$ and define a morphism $\varphi: \{a_i, b_i | 1 \leq i \leq k\}^* \to \{A_i, B_i | 1 \leq i \leq k\}^*$ in the natural way (e.g. the morphism induced by $\varphi(a_i) = A_i$ and $\varphi(b_i) = B_i$). Now, for a word $w = w_1w_2 \cdots w_k \in \mathcal{L}$, note that $w_j \in \{a_j, b_j\}$ for $1 \leq j \leq k$. Define that $\mathfrak{v}(a_i) = x_i$ and $\mathfrak{v}(b_i) = 0$. In this case, we see that (due to the structure of $A_i$ and $B_i$):
$$u^T\varphi(w_1 w_2 \cdots w_k)v = \frac{1}{\sum_{j = 1}^{k}(x_j+1)} \sum_{\ell = 1}^{k} \mathfrak{v}(w_\ell)
$$ 
Note of course that the factor $\frac{1}{\sum_{j = 1}^{k}(x_j+1)}$ is the same for any $w \in \mathcal{L}$.

Assume then that there exists two words $\alpha, \beta \in \mathcal{L}$ with $\alpha \neq \beta$ such that $u^T\varphi(\alpha)v = u^T\varphi(\beta)v$ (i.e. assume that $\pfa$ is not injective). Then $\sum_{\ell = 1}^{k} \mathfrak{v}(\alpha_\ell) = \sum_{i \in S_1}^{k} x_i = \sum_{i \in S_2}^{k} x_i = \sum_{\ell = 1}^{k} \mathfrak{v}(\beta_\ell)$, where $S_1 = \{x_i ; |\alpha|_{a_i}>0\}$ and $S_2 = \{x_i ; |\beta|_{a_i}>0\}$. This is true if and only if the instance $S$ of the equal subset sum problem has a solution as required (note that only the empty set has a sum of zero which has unique representation $b_1\cdots b_k$). Since $A_i$ and $B_i$ are upper-triangular, with initial state $1$ and final state $2$, then $\pfa$ is linearly ambiguous.
\end{proof}

\section{Conclusion}

There are a variety of open problems remaining. For example, does Theorem~\ref{mainThm} still hold for quadratic ambiguity, when taken alongside the other constraints (letter monotonic language and commutative matrices). Another direction is to improve the complexity lower bound of Theorem~\ref{nphardthm} to show it is either PSPACE-hard, EXPSPACE-hard or undecidable, under the same constraints as in the theorem statement.

\bibliography{refs}

\end{document}